\documentclass[a4paper,fleqn]{article}
\usepackage{a4wide}
\usepackage{amsmath,amssymb,amsthm}
\usepackage{thmtools}
\usepackage{enumerate}
\usepackage{abstract}
\usepackage{multirow}
\usepackage{tikz}
\usepackage{tikz-qtree} % Strategiapuita varten.
\tikzset{grow'=right} % Make trees grow from left to right.
\tikzset{every tree node/.style={anchor=base west}} % Align nodes of the tree to the left (west).
\usepackage{hyperref}
\usepackage[T1]{fontenc}
\usepackage{mathpazo}
\makeatletter\@ifpackageloaded{mathpazo}\@tempswatrue\@tempswafalse
\if@tempswa
  \DeclareFontFamily{OT1}{pzc}{}
  \DeclareFontShape{OT1}{pzc}{m}{it}{<-> s * [1.15] pzcmi7t}{}
  \DeclareMathAlphabet{\mathpzc}{OT1}{pzc}{m}{it}
\fi\makeatother
\usepackage[utf8]{inputenc}

\usepackage[sorting=nyt,giveninits=true,backend=biber,mincrossrefs=99,maxbibnames=4]{biblatex}
\makeatletter\@ifpackageloaded{biblatex}{%
  \usepackage{csquotes} % Silence a warning.
  \bibliography{../../references}
  \renewbibmacro{in:}{%
    \ifentrytype{incollection}{\printtext{\bibstring{in}\intitlepunct}}{}}
  \renewbibmacro{publisher+location+date}{%
    \iflistundef{publisher}
      {\setunit*{\addcomma\space}}
      {\setunit*{\addcomma\space}}%
    \printlist{publisher}%
    \setunit*{\addcomma\space}%
    \printlist{location}%
    \setunit*{\addcomma\space}%
    \usebibmacro{date}%
    \newunit}
  \DeclareFieldFormat[article]{pages}{#1\isdot}
  \DeclareFieldFormat[article,incollection,inproceedings,unpublished,eprint]{title}{#1\isdot}
  \DeclareFieldFormat[thesis]{title}{\mkbibemph{#1\isdot}}
  \DeclareFieldFormat[unpublished]{date}{(#1)\isdot}
  \DeclareFieldFormat[unpublished]{note}{#1\nopunct} % tämä poistaa turhan pisteen ennen vuosilukua
  \DeclareFieldFormat[eprint]{date}{(#1)\isdot}
  \DeclareFieldFormat[eprint]{note}{#1\nopunct} % tämä poistaa turhan pisteen ennen vuosilukua
  \DeclareFieldFormat[article]{journaltitle}{\mkbibemph{#1\isdot}}
  \AtEveryBibitem{%
    \ifentrytype{book}{}{% Remove editor except for books
      \clearname{editor}
    }
  }
  \newbibmacro*{bbx:parunit}{%
    \ifbibliography
      {\setunit{\bibpagerefpunct}\newblock
       \usebibmacro{pageref}%
       \clearlist{pageref}%
       \setunit{\adddot\par\nobreak}}
      {}
  }
  \renewbibmacro*{doi+eprint+url}{%
    \usebibmacro{bbx:parunit}% Added
    \iftoggle{bbx:doi}
      {\printfield{doi}}
      {}%
    \iftoggle{bbx:eprint}
      {\usebibmacro{eprint}}
      {}%
    \iftoggle{bbx:url}
      {\usebibmacro{url+urldate}}
      {}
  }
  \renewbibmacro*{eprint}{%
    \usebibmacro{bbx:parunit}% Added
    \iffieldundef{eprinttype}
      {\printfield{eprint}}
      {\printfield[eprint:\strfield{eprinttype}]{eprint}}
  }
  \renewbibmacro*{url+urldate}{%
    \usebibmacro{bbx:parunit}% Added
    \printfield{url}%
    \iffieldundef{urlyear}
      {}
      {\setunit*{\addspace}%
       \printtext[urldate]{\printurldate}}
  }
  
}{}\makeatother

\declaretheorem[numberwithin=section,refname={theorem,theorems},Refname={Theorem,Theorems}]{theorem}
\declaretheorem[sibling=theorem,style=definition]{definition}
\declaretheorem[sibling=theorem,name=Example,style=definition]{example}
\declaretheorem[sibling=theorem,name=Lemma]{lemma}
\declaretheorem[sibling=theorem,name=Proposition]{proposition}

\makeatletter\@ifpackageloaded{hyperref}{%
  \usepackage{xcolor}
  \definecolor{dark-red}{rgb}{0.4,0.15,0.15}
  \definecolor{dark-blue}{rgb}{0.15,0.15,0.4}
  \definecolor{medium-blue}{rgb}{0,0,0.5}
  \hypersetup{
    colorlinks,
    linkcolor={dark-red},
    citecolor={dark-blue},
    urlcolor={medium-blue}%
  }
}{}\makeatother

\newcommand{\address}[1]{\vspace{-2em}\begin{center}{\footnotesize #1}\end{center}}

\usepackage{sectsty}
\sectionfont{\large\bfseries}
\subsectionfont{\normalsize\bfseries}

\DeclareMathOperator{\pref}{pref}
\DeclareMathOperator{\suff}{suff}
\providecommand{\abs}[1]{\lvert#1\rvert}

\newcommand{\infw}[1]{\mathbf{#1}}
\newcommand{\N}{\mathbb{N}}
\newcommand{\Z}{\mathbb{Z}}

\newcommand{\Lang}[2][]{\mathcal{L}_{#1}(#2)}
\newcommand{\del}{\partial}

\newcommand{\keywords}[1]{\par\noindent{\footnotesize{\em Keywords\/}: #1}}

\begin{document}
  \title{On winning shifts of marked uniform substitutions}
  \author{Jarkko Peltomäki and Ville Salo\\
          \small \href{mailto:r@turambar.org}{r@turambar.org}, \quad \href{mailto:vosalo@utu.fi}{vosalo@utu.fi}}
  \date{}
  \maketitle
  \address{Turku Centre for Computer Science TUCS, Turku, Finland\\
  University of Turku, Department of Mathematics and Statistics, Turku, Finland}

  \noindent
  \hrulefill
  \begin{abstract}
    \vspace{-1em}
    \noindent

    \noindent
    The second author introduced with I. Törmä a two-player word-building game [Playing with Subshifts, Fund. Inform.
    132 (2014), 131--152]. The game has a predetermined (possibly finite) choice sequence $\alpha_1$, $\alpha_2$,
    $\ldots$ of integers such that on round $n$ the player $A$ chooses a subset $S_n$ of size $\alpha_n$ of some fixed
    finite alphabet and the player $B$ picks a letter from the set $S_n$. The outcome is determined by whether the word
    obtained by concatenating the letters $B$ picked lies in a prescribed target set $X$ (a win for player $A$) or not
    (a win for player $B$). Typically, we consider $X$ to be a subshift. The winning shift $W(X)$ of a subshift $X$ is
    defined as the set of choice sequences for which $A$ has a winning strategy when the target set is the language of
    $X$. The winning shift $W(X)$ mirrors some properties of $X$. For instance, $W(X)$ and $X$ have the same entropy.
    Virtually nothing is known about the structure of the winning shifts of subshifts common in combinatorics on words.
    In this paper, we study the winning shifts of subshifts generated by marked uniform substitutions, and show that
    these winning shifts, viewed as subshifts, also have a substitutive structure. Particularly, we give an explicit
    description of the winning shift for the generalized Thue-Morse substitutions. It is known that $W(X)$ and $X$
    have the same factor complexity. As an example application, we exploit this connection to give a simple derivation
    of the first difference and factor complexity functions of subshifts generated by marked substitutions. We describe
    these functions in particular detail for the generalized Thue-Morse substitutions.

    \vspace{1em}
    \keywords{two-player game, winning shift, marked substitution, factor complexity, generalized Thue-Morse word}
    \vspace{-1em}
  \end{abstract}
  \hrulefill

  \section{Introduction}
  In the paper \cite{2014:playing_with_subshifts}, the second author introduced with I. Törmä a two-player
  word-building game. The two players, Alice and Bob, agree on a finite alphabet $S$, a target set $X$ of words over
  $S$, game length $n \in \N \cup \{\N\}$, and a choice sequence $\alpha_1 \alpha_2 \cdots \alpha_n$ (a word) of
  integers in $\{1,2,\ldots,\abs{S}\}^n$. On the round $j$ of the game, $1 \leq j \leq n$, Alice first chooses a subset
  $S_j$ of $S$ of size $\alpha_j$ and then Bob picks a letter $a_j$ from the subset $S_j$. During the game, Alice and
  Bob thus together build the word $a_1 a_2 \cdots a_n$ (finite or infinite). If this built word is in the target set
  $X$, then Alice wins, otherwise Bob does. In other words, Alice aims to build a valid word of $X$ while her adversary
  Bob attempts to introduce a forbidden word.
  
  In studying games of this sort, it would be typical to fix a choice sequence and see what conditions on $X$ guarantee
  the existence of a winning strategy for one of the players. The work of \cite{2014:playing_with_subshifts} adopts the
  opposite point of view: fix a set $X$ and see for which choice sequences Alice has a winning strategy. This set of
  choice sequences, dubbed as the winning set $W(X)$ of $X$, turns out to be a very interesting object. First of all,
  if $X$ is a subshift, then $W(X)$, now called the winning shift of $X$, is also a subshift, and the set of factors of
  $W(X)$ of length $k$ is exactly the winning set of factors of $X$ of length $k$. Actually the winning set $W(X)$
  inherits many properties of $X$. For instance, if $X$ is a regular language, so is $W(X)$, and if $X$ computable,
  then so is $W(X)$. The most interesting result, which sparked the research in this paper, is the fact that the sets
  $X$ and $W(X)$ have the same cardinality so, for a subshift $X$, the winning shift $W(X)$ has the same entropy and
  factor complexity function as $X$. Now the winning set $W(X)$ is in a sense simpler than $X$ because it is downward
  closed: if any letter of a choice sequence in $W(X)$ is downgraded to a smaller letter, then the resulting word is
  still in $W(X)$. The winning set $W(X)$ is thus a rearrangement of $X$ to a downward closed set. Indeed, the winning
  set can be significantly simpler: for instance, the winning set of a Sturmian subshift is the subshift over $\{1,2\}$
  whose words contain at most one letter $2$.

  Descriptions of the winning shifts for particular subshifts remain largely unknown. In this work, we provide such
  descriptions for the winning shifts of subshifts generated by marked uniform substitutions. A marked substitution is
  a substitution such that all images of letters begin with distinct letters and end with distinct letters. We prove
  that all long enough choice sequences in such a winning shift are obtained from a few core choice sequences by
  substitution (\autoref{thm:marked_substitutive}). Let us make this more precise. Let $\tau\colon S^* \to S^*$ be a
  marked uniform substitution of length $M$, and let $w$ be a short choice sequence in the language of the winning
  shift $W(\tau)$ of the subshift generated by $\tau$. Write $w = \diamond u a$ for letters $\diamond$ and $a$. Then
  $z \sigma(u) a$ is in the language of $W(\tau)$; here $\sigma$ is the substitution defined by $\sigma(k) = k1^{M-1}$
  and the word $z$ is in the winning set of certain suffixes of the $\tau$-images of a subset of $S$ of size
  $\diamond$. All long enough choice sequences in the language of $W(\tau)$ are essentially obtained in this way. In
  general, the short choice sequences and possible words $z$ can be very complex and they elude any simple description,
  but they can be efficiently computed. This together with \autoref{thm:marked_substitutive} allows us to rapidly
  compute the language of the winning shift $W(\tau)$. If we make additional assumptions on $\tau$, then the situation
  can be simplified. For instance, if $\tau$ is permutive (letters at a fixed position of the $\tau$-images form a
  permutation of the alphabet $S$), then $z$ is simply of the form $\diamond 1^i$ for some $i$ such that $0 \leq i < M$
  (\autoref{prp:permutive}). This class of permutive uniform substitutions includes the generalized Thue-Morse
  substitutions. For them, we compute all involved parameters and give full description of the whole winning shift
  (\autoref{sec:gtm}).

  The structure of the winning shift of a marked uniform substitution is quite easy to comprehend, and we apply our
  results to give a simple derivation of the first difference function of such a substitution
  (\autoref{thm:first_difference}). This function can in turn be used to derive the factor complexity function. A. Frid
  has derived these functions previously with other methods \cite{1998:on_uniform_d0l_words}; see also
  \cite{1996:reconnaissabilite_des_substitutions_et_complexity_des_suites_automatiques}. Our arguments and Frid's
  arguments, which by the way apply in a more general setting, in the end reduce to the same fundamental observations,
  but the high-level view is completely different. We prefer gaming and feel that analyzing Alice and Bob's match is
  fresh and, more importantly, fun. The aim of this paper is to describe the winning shift; the connection to factor
  complexity is more of a motive for the study, a curiosity. We do, however, derive the factor complexity function in
  full detail for the generalized Thue-Morse words, just as we describe their winning shifts completely
  (\autoref{sec:gtm}). These complexity functions have been derived in full generality previously by \v{S}. Starosta in
  \cite{2012:generalized_thue-morse_words_and_palindromic_richness} using an intriguing connection to so-called
  $G$-rich words. Results in specialized cases were known before Starosta, see
  \cite{1989:enumeration_of_factors_in_the_thue-morse_word,1989:some_combinatorial_properties_of_the_thue-morse_sequence,1995:subword_complexity_of_the_generalized_thue-morse_word}.
  A short version of this paper with results applying only to the generalized Thue-Morse words was presented in the
  proceedings of RuFiDiM IV \cite{2017:on_winning_shifts_of_generalized_thue-morse_substitutions}.

  The paper is organized as follows. In the next section, we give the necessary definitions and results needed. After
  this in \autoref{sec:tm_example}, we outline the structure of the winning shift of the Thue-Morse substitution and
  use it as a motivating example to introduce our ideas. \autoref{sec:main_results} contains the main results. We show
  that generally short choice sequences can be substituted to obtain longer choice sequences, but the additional
  assumption of markedness is needed for desubstitution. We end \autoref{sec:main_results} by deriving a recurrence for
  the first difference function of a marked uniform substitution. The final section is devoted to the generalized
  Thue-Morse substitutions. We completely describe their winning shifts and, as an application, derive formulas for
  their factor complexity functions.

  \section{Notation and Preliminary Results}\label{sec:preliminaries}
  \subsection{Standard Definitions}
  Here we briefly define word-combinatorial notions; further details are found in, e.g.,
  \cite{2002:algebraic_combinatorics_on_words}. An alphabet $S$ is a nonempty finite set of letters, and we denote by
  $S^*$ the set of finite words over $S$. The set of words over $S$ of length $n$ is denoted by $S^n$, and by
  $S^{\leq n}$ we denote the set of words over $S$ with length at most $n$. Infinite words over $S$ are sequences in
  $S^\N$. The length of a finite word $w$ is denoted by $\abs{w}$, and the empty word $\varepsilon$ is the unique word
  of length $0$. Suppose that $w$ is a word (finite or infinite) such that $w = uzv$ for some words $u$, $z$, and $v$.
  Then we say that $z$ is a factor of $w$. If $u = \varepsilon$ (respectively $v = \varepsilon)$, then we call the
  factor $z$ a prefix (respectively suffix) of $w$. If $u = \varepsilon$ and $z \neq w$, then $z$ is a proper prefix of
  $w$; similarly we define a proper suffix of $w$. We say that $z$ occurs at position $\abs{u}$ of $w$; the position
  $\abs{u}$ is an occurrence of the factor $z$. Thus we index letters from $0$. The word $\del_{i,j}(w)$, where
  $i + j \leq \abs{w}$, is obtained from the word $w$ by deleting $i$ letters from the beginning and $j$ letters from
  the end. An infinite word is ultimately periodic if it is of the form $uvvv\cdots$; otherwise it is aperiodic.

  A \emph{subshift} $X$ is a subset of $S^\N$ defined by some set $F$ of forbidden words:
  \begin{equation*}
    X = \{w \in S^\N\colon \text{no word of $F$ occurs in w}\}.
  \end{equation*}
  We denote by $\Lang[X]{n}$ the set of words of length $n$ occurring in words of $X$ and define the \emph{language}
  $\Lang{X}$ of $X$ as the set $\bigcup_{n \in \N} \Lang[X]{n}$. The subshift $X$ is uniquely defined by its language.
  The function $f$ defined by letting $f(n) = \abs{\Lang[X]{n}}$ is called the \emph{factor complexity function of $X$}
  (we assume that $X$ is known from context), and it counts the number of words of length $n$ in the language of $X$.
  We define the \emph{first difference function $\Delta$} by setting $\Delta(n) = f(n) - f(n-1)$ and $\Delta(0) = 1$.
  This function measures the growth of the factor complexity function.

  \subsection{Substitutions}
  A function $\tau\colon S^* \to S^*$ is a called a \emph{substitution} if $\tau(uv) = \tau(u)\tau(v)$ for all
  $u, v \in S^*$. In this paper, we typically select $S = \{0, 1, \ldots, \abs{S} - 1\}$. If $\tau(s)$ has the same
  length for every $s \in S$, then we say that $\tau$ is \emph{uniform}. In this paper, we assume that for uniform
  substitutions we have $\abs{\tau(s)} \geq 2$ for all $s \in S$. We call the images of letters, the words $\tau(s)$,
  \emph{$\tau$-images}. If $\tau(s)$ begins with $s$ and $\lim_{n\to\infty} \abs{\tau^n(s)} = \infty$ for a letter $s$,
  then the infinite word obtained by repeatedly applying $\tau$ to $s$, denoted by $\tau^\omega(s)$, is a \emph{fixed
  point} of the substitution $\tau$. Consider the language $\mathcal{L}$ defined as the set
  \begin{equation*}
    \bigcup_{s \in S} \{w \in S^*\colon \text{$w$ occurs in $\tau^n(s)$ for some $n \geq 0$}\}
  \end{equation*}
  consisting of the factors of the words obtainable by applying $\tau$ repeatedly to the letters of $S$. Let
  \begin{equation*}
    \Lang{\tau} = \{w \in \mathcal{L}\colon \text{there exists arbitrarily long words $u$ and $v$ such that $uwv \in \mathcal{L}$}\}.
  \end{equation*}
  The subshift generated by $\tau$ is simply the subshift with the language $\Lang{\tau}$ (i.e., we forbid the
  complement of $\Lang{\tau}$). The substitution $\tau$ is \emph{primitive} if there is an integer $n$ such that
  $\tau^n(s)$ contains all letters of $S$ for every $s \in S$. The substitution $\tau$ is \emph{aperiodic} if the
  subshift generated by $\tau$ does not contain ultimately periodic infinite words. We assume that all substitutions
  considered are aperiodic.
  
  We call a substitution $\tau$ \emph{left-marked} if all of its $\tau$-images begin with distinct letters. In other
  words, there exists a permutation $\pi\colon S \to S$ such that $\tau(k) = \pi(k) w_k$ for $k \in S$. Analogously we
  define \emph{right-marked} substitutions. If a substitution is left-marked and right-marked, then it is simply called
  a \emph{marked} substitution. Observe also that marked substitutions have an obvious but important property: if a
  single letter of a $\tau$-image is changed, then the resulting word is no longer a valid $\tau$-image. A substitution
  is \emph{permutive} if there exists permutations $\pi_1$, $\pi_2$, $\ldots$, $\pi_M$ from $S$ to $S$ such that
  $\tau(k) = \pi_1(k)\pi_2(k) \cdots \pi_M(k)$ for $k \in S$. A permutive substitution is uniform and marked.

  We say that a word $w$ in $\Lang{\tau}$ admits an interpretation $(a_0 \cdots a_{n+1}, i, j)$ for letters
  $a_0$, $\ldots$, $a_{n+1}$ by $\tau$ if $w = \del_{i,j}(\tau(a_0 \cdots a_{n+1}))$, $0 \leq i < \abs{\tau(a_0)}$,
  $0 \leq j < \abs{\tau(a_{n+1})}$, and $a_0 \cdots a_{n+1} \in \Lang{\tau}$. The word $a_0 \cdots a_{n+1}$ is called
  an \emph{ancestor} of the word $w$. We say that $(u_1, u_2)$ is a \emph{synchronization point} of $w$ (for $\tau$)
  if $w = u_1 u_2$ and whenever $v_1 w v_2 = \tau(z)$ for some $z \in \Lang{\tau}$ and some words $v_1$ and $v_2$, then
  $v_1 u_1 = \tau(t_1)$ and $u_2 v_2 = \tau(t_2)$ for some words $t_1$ and $t_2$ such that $z = t_1 t_2$. We say that
  $\tau$ has \emph{synchronization delay} $L$ if every word in $\Lang{\tau}$ of length at least $L$ has at least one
  synchronization point and $L$ is minimal. Observe that if $\tau$ is marked, then all words in $\Lang{\tau}$ of length
  at least $L$ have a unique ancestor. We assume that all substitutions considered in this paper have a synchronization
  delay. It follows from a theorem of Mossé
  \cite[Corollaire~3.2.]{1992:puissances_de_mots_et_reconnaissabilitie_des_points_fixes} that the synchronization delay
  of a uniform, primitive, and aperiodic substitution always exists.\footnote{Mossé's Theorem applies to any primitive and aperiodic substitution.}

  Let $\tau$ be a uniform substitution of length $M$ with synchronization delay $L$. Let $w$ in $\Lang{\tau}$ be a word
  such that $\abs{w} \geq L$. Suppose that $w$ has an ancestor $z$, so that $w = \del_{i,j}(\tau(z))$ with
  $0 \leq i, j < M$. While $w$ might have several ancestors, the uniformity of $\tau$ and the fact that $w$ has at
  least one synchronization point ensure that the numbers $i$ and $j$ are independent of the chosen ancestor $z$. In
  fact, the positions $i$ and $j$ mark a synchronization point of $w$. All in all, the number $i$ uniquely identifies
  the positions of $w$ where the $\tau$-images of the letters of any ancestor of $w$ begin at, and we say that $w$ has
  \emph{decomposition $i \bmod{M}$}.

  \subsection{Word Games}
  Next we define precisely the word game in which two players, Alice and Bob, build a finite or infinite word. A
  \emph{word game} is a quadruple $(S, n, X, \alpha)$, where $S$ is an alphabet, $n \in \N \cup \{\N\}$, the
  \emph{target set} $X$ is a subset of $S^n$, and the \emph{choice sequence} $\alpha$ is a word of length $n$ (an
  infinite word if $n = \N$) over the alphabet $\{1,2,\ldots,\abs{S}\}$. We may allow the target set $X$ to contain
  words of distinct lengths by using $X \cap S^n$ in place of $X$; this will always be clear from context.
  
  Denote by $G$ the word game $(S, n, X, \alpha)$ with $n \in \N$, and write $\alpha = \alpha_1 \cdots \alpha_n$ for
  letters $\alpha_i$. During the round $i$, $1 \leq i \leq n$, of this game, first Alice chooses a subset $S_i$ of $S$
  of size $\alpha_i$. Then Bob picks a letter $a_i$ from the set $S_i$. After $n$ rounds, Alice and Bob have together
  built a word $a_1 a_2 \cdots a_n$. If $a_1 a_2 \cdots a_n \in X$, then Alice wins the game $G$ and otherwise Bob
  does. An example is provided at the beginning of \autoref{sec:tm_example}, and more examples are found in
  \cite{2014:playing_with_subshifts}. The notions presented in this paragraph extend to the case $n = \N$ in a natural
  way.
  
  Alice's \emph{strategy} for $G$ is a function $s\colon S^{\leq i} \to 2^S$ that specifies which subset of size
  $\alpha_{i+1}$ she should choose next given the word of length $i$ constructed so far. Similarly we define Bob's
  strategy as a partial function $s\colon S^{\leq i} \times 2^S \to S$ specifying which letter Bob should pick given
  the word constructed so far and the subset chosen by Alice. Let $s_A$ and $s_B$ respectively be Alice's strategy and
  Bob's strategy for the game $G$. The \emph{play} $p(G, s_A, s_B)$ of the strategy pair $(s_A, s_B)$ is the word
  $a_1 a_2 \cdots a_n$ defined inductively by $a_{i+1} = s_B(a_1 \cdots a_i, s_A(a_1 \cdots a_i))$ with
  $a_1 \cdots a_0 = \varepsilon$ (if $n = \N$, then the play $a_1 a_2 \cdots$ is simply infinite). We say that Alice's
  strategy $s$ is \emph{winning} if $p(G, s, s_B) \in X$ for all Bob's strategies $s_B$ (Alice wins no matter how Bob
  plays). Analogously Bob's strategy $s$ is winning if $p(G, s_A, s) \notin X$ for all Alice's strategies $s_A$. If
  $n \in \N$ or $X$ is a closed set in the product topology of $S^\N$ (in particular, if $X$ is a subshift), then a
  winning strategy always exists for one of the players \cite{1953:infinite_games_with_perfect_information}. In this
  paper, we consider Bob's strategies only indirectly. Thus whenever we talk about a winning strategy we mean that it
  is Alice's winning strategy. Similarly by a \emph{winning play} we mean a play by a strategy pair $(s_A, s_B)$ where
  $s_A$ is Alice's winning strategy.
  
  As mentioned in the introduction, we are interested in the choice sequences for which Alice has a winning strategy.
  Given a subset $X$ of $S^n$, where $n \in \N \cup \{\N\}$, we define the \emph{winning set} $W(X)$ of $X$ as the set
  \begin{equation*}
    \{\alpha \in \{1,\ldots,\abs{S}\}^n \colon \text{Alice has a winning strategy for the word game $(S, n, X, \alpha)$}\}.
  \end{equation*}
  Notice that in general Alice has several winning strategies for a choice sequence in $W(X)$ We often omit the
  alphabet $S$, it will be clear from the context. For a language $X \subseteq S^*$, we set
  \begin{equation*}
    W(X) = \bigcup_{n \in \N} W(X \cap S^n)
  \end{equation*}
  and call also this set the winning set of $X$. If $n = \N$ and $X$ is a subshift, then we call $W(X)$ the
  \emph{winning shift} of $X$; if the subshift $X$ is generated by a substitution $\tau$, then we denote its winning
  shift by $W(\tau)$. Indeed, in \cite[Proposition~3.4]{2014:playing_with_subshifts}, the following result was obtained.

  \begin{proposition}
    If $X$ is a subshift, then $W(X)$ is a subshift and $\Lang{W(X)} = W(\Lang{X})$.
  \end{proposition}

  We abuse notation and write $W(X)$ for $\Lang{W(X)}$, it is always clear from context whether we consider finite
  words or infinite words. In addition, we have the following observation.

  \begin{lemma}\label{lem:inclusion}
    Let $X$ and $Y$ be sets containing words of equal length. If $X \subseteq Y$, then $W(X) \subseteq W(Y)$.
  \end{lemma}
  \begin{proof}
    Alice's winning strategy for a word game with target set $X$ and choice sequence in $W(X)$ is sufficient as it is
    for her to win in the game with the same choice sequence and target set $Y$.
  \end{proof}

  We endow the alphabet $\{1,\ldots,\abs{S}\}$ with the natural order $1 < 2 < \ldots < \abs{S}$. Suppose that $u$ and
  $v$ are words over this alphabet (finite or infinite), and write $u = u_0 \cdots u_{n-1}$ and
  $v = v_0 \cdots v_{m-1}$ for letters $u_i$, $v_i$. Then we write $u \leq v$ if and only if $n = m$ and $u_i \leq v_i$
  for $i = 0, \ldots, n - 1$. The winning set $W(X)$ is \emph{downward closed} with respect to this partial ordering:
  if $u \leq v$ and $v \in W(X)$, then $u \in W(X)$. This is simply because \emph{downgrading} a letter from the choice
  sequence only makes Bob's chances of winning slimmer.

  Observe that the winning strategies for finite choice sequences ending with the letter $1$ are just trivial
  extensions of winning strategies of shorter choice sequences ending with a letter greater than $1$. Thus we say that
  a finite choice sequence is \emph{reducible} if it ends with $1$ and \emph{irreducible} otherwise. The infinite words
  of the winning shift $W(X)$ are obtainable from irreducible choice sequences by appending infinitely many letters $1$
  and by taking closure. A rule of thumb for the rest of the paper is that to describe the structure of the winning
  sets it is enough to study only irreducible choice sequences.

  Finally, we need the next proposition \cite[Proposition~5.7]{2014:playing_with_subshifts} that motivates the
  presented results.

  \begin{proposition}\label{prp:cardinality}
    If $n \in \N$ and $X \subseteq S^n$, then $\abs{W(X)} = \abs{X}$.
  \end{proposition}

  We note that a subset $W$ of $\{0,1\}^n$ can be interpreted as a family of subsets of $\{1,2,\ldots,n\}$ (a
  so-called set system) by considering a word $w \in \{0,1\}^n$ as the characteristic function of a subset.
  \autoref{prp:cardinality} has been proven in relation to set systems in
  \cite{2002:shattering_news}.\footnote{Formally, the result of \cite{2002:shattering_news} corresponds to the binary case of \autoref{prp:cardinality}. Their \emph{order-shattered sets} for the set system whose characteristic functions are $X \subseteq \{0,1\}^n$ correspond to the choice sequences in $W(X^R)^R$, where $R$ is word reversal, that is, their games are played from right to left.}

  \section{The Motivating Example of the Thue-Morse Substitution}\label{sec:tm_example}
  In this section, we consider the winning shift of the Thue-Morse substitution. Through examples, we
  describe the substitutive structure of this winning shift and outline how it can be used to compute the factor
  complexity of the subshift generated by the Thue-Morse substitution. Our claims are rigorously derived in the
  subsequent sections in a more general setting.

  \begin{table}
    \centering
    \begin{tabular}{|c|l||c|l||c|l|}
      \hline
      $n$ & & $n$ & & $n$ & \\
      \hline
      $1$ & $\diamond$         & $9$ & $\diamond 11111112$         & $17$ & $\diamond 1111111111111112$ \\
      \hline
      $2$ & $\diamond 2$       & $10$ & $\diamond 111111112$       & $18$ & $\diamond 11111111111111112$ \\
          &                    &      & $\diamond 211111112$       &      & $\diamond 21111111111111112$ \\
      \hline
      $3$ & $\diamond 12$      & $11$ & $\diamond 1111111112$      & $19$ & $\diamond 111111111111111112$ \\
          &                    &      & $\diamond 1211111112$      &      & $\diamond 121111111111111112$ \\
      \hline
      $4$ & $\diamond 112$     & $12$ & $\diamond 11111111112$     & $20$ & $\diamond 1111111111111111112$ \\
          & $\diamond 212$     &      & $\diamond 11211111112$     &      & $\diamond 1121111111111111112$ \\
      \hline
      $5$ & $\diamond 1112$    & $13$ & $\diamond 111111111112$    & $21$ & $\diamond 11111111111111111112$ \\
          &                    &      & $\diamond 111211111112$    &      & $\diamond 11121111111111111112$ \\
      \hline
      $6$ & $\diamond 11112$   & $14$ & $\diamond 1111111111112$   & $22$ & $\diamond 111111111111111111112$ \\
          & $\diamond 21112$   &      &                            &      & $\diamond 111121111111111111112$ \\
      \hline
      $7$ & $\diamond 111112$  & $15$ & $\diamond 11111111111112$  & $23$ & $\diamond 1111111111111111111112$ \\
          & $\diamond 121112$  &      &                            &      & $\diamond 1111121111111111111112$ \\
      \hline
      $8$ & $\diamond 1111112$ & $16$ & $\diamond 111111111111112$ & $24$ & $\diamond 11111111111111111111112$ \\
          &                    &      &                            &      & $\diamond 11111121111111111111112$ \\
      \hline
    \end{tabular}
    \caption{The irreducible choice sequences of the winning shift of the Thue-Morse substitution for lengths $1$ to
    $24$. The letter $\diamond$ can be substituted by both of the letters $1$ and $2$.}
    \label{tbl:thue_morse}
  \end{table}

  Let $\tau$ be the Thue-Morse substitution: $\tau(0) = 01$, $\tau(1) = 10$. The substitution $\tau$ is uniform,
  primitive, and marked, and it is readily proven that it is aperiodic. With an exhaustive search, it is easily
  established that its synchronization delay is $4$ (see also \autoref{lem:synch_delay}). The fixed point
  \begin{equation*}
    \tau^\omega(0) = 01101001100101101001011001101001100101100110100101101001 \cdots
  \end{equation*}
  is the famous Thue-Morse word, which is overlap-free (i.e., it does not contain a factor of the form
  $auaua$ for a word $u$ and a letter $a$). For more details on the substitution $\tau$, see for example
  \cite[Section~2.2]{1983:combinatorics_on_words}.
  
  In \autoref{tbl:thue_morse}, we list irreducible choice sequences of $W(\tau)$ for lengths $1$ to
  $24$.\footnote{Here we indeed abuse notation, and we should write $\Lang{W(\tau)}$ for $W(\tau)$. Remember that reducible choice sequences of length $n$ are obtained by padding shorter irreducible choice sequences with the letter $1$.}
  For the choice sequence $2212$, Alice has the following winning strategy:
  \begin{align*}
    \varepsilon        &\mapsto \{0, 1\}, \\
    0, 1               &\mapsto \{0, 1\}, \\
    00, 10             &\mapsto \{1\}, \\
    01, 11             &\mapsto \{0\}, \\
    001, 101, 010, 110 &\mapsto \{0, 1\},
  \end{align*}
  the other arguments being irrelevant. This strategy is depicted in \autoref{fig:thue-morse_strategy} as a
  \emph{strategy tree}; this tree representation is used throughout this paper. Whenever Alice has more than one choice
  according to her strategy, the tree branches to several nodes that correspond to Alice's possible choices of letters.
  We omit edges from the tree when there are no branchings.

  \begin{figure}
    \centering
    \begin{minipage}{.5\textwidth}
      \centering
      \begin{tikzpicture}
        \tikzset{level distance=40pt}
        \Tree [.$\varepsilon$ [.$0$ [.$01$ [.$0$ ] [.$1$ ] ] [.$10$ [.$0$ ] [.$1$ ] ] ] [.$1$ [.$01$ [.$0$ ] [.$1$ ] ] [.$10$ [.$0$ ] [.$1$ ] ] ] ]
      \end{tikzpicture}
    \end{minipage}%
    \begin{minipage}{.5\textwidth}
      \centering
      \begin{tikzpicture}
        \tikzset{level distance=40pt}
        \tikzset{level 3/.style={level distance=50pt}}
        \Tree [.$\varepsilon$ [.$01$ [.$0110$ [.$01$ ] [.$10$ ] ] [.$1001$ [.$01$ ] [.$10$ ] ] ] [.$10$ [.$0110$ [.$01$ ] [.$10$ ] ] [.$1001$ [.$01$ ] [.$10$ ] ] ] ]
      \end{tikzpicture}
    \end{minipage}
    \caption{Winning strategies for Alice for the choice sequences $2212$ and $21211121$ in the case of the Thue-Morse
    substitution.}\label{fig:thue-morse_strategy}
  \end{figure}
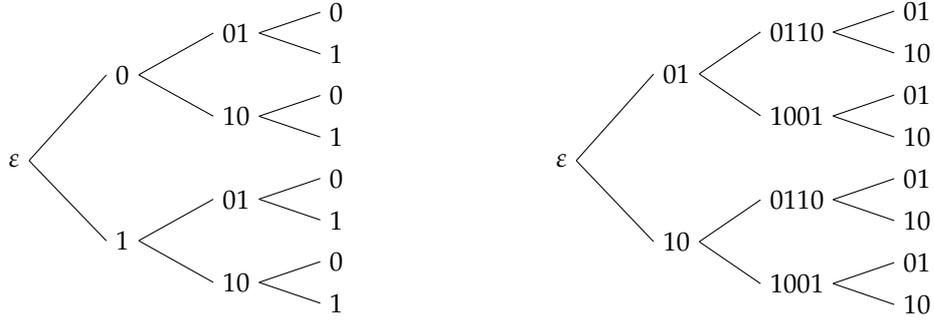

  \autoref{tbl:thue_morse} contains many patterns. By \autoref{prp:cardinality}, the number of irreducible choice
  sequences of length $n$ is counted by the first difference function $\Delta(n)$. Based on the data, it seems that
  $\Delta(n) \in \{2, 4\}$ for all $n \geq 1$ and $\Delta(n) = 4$ only if $n = 2^k + \ell + 1$ for $k \geq 1$ and
  $1 \leq \ell \leq 2^{k-1}$. This is of course readily observed when looking at the factor complexity function; here
  we see much more: the rule described next confirms the preceding observations.

  We observe that a choice sequence $\alpha$ in the winning shift always seems to contain at most three occurrences of
  $2$. Moreover, if $\alpha$ contains exactly three occurrences of $2$, then the distance between the two final
  occurrences is $2^k - 1$ for some $k \geq 1$, and the middle occurrence is preceded by at most $2^{k-1}$ occurrences
  of the letter $1$. The rule seems to be the following. If $n = 3 \cdot 2^k + 2$, then the only irreducible choice
  sequence of length $n$ (up to the difference at the very beginning) is $\smash[t]{\diamond 1^{3 \cdot 2^k} 2}$. Then
  the number of $1$s increases until there are $2^{k+2} - 1$ of them. Next a third occurrence of $2$ can be introduced:
  the choice sequences of length $2^{k+2} + 2$ are $\diamond 2 1^{2^{k+2} - 1} 2$ and $\diamond 1^{2^{k+2}} 2$ (the
  former choice sequence downgraded). Then the number of $1$s before the second to last occurrence of $2$ starts to
  grow one by one until the choice sequences considered are of length $3 \cdot 2^{k+1} + 1$, and then the pattern
  repeats. The observed rule suggests that irreducible choice sequences of $W(\tau)$ of lengths $2^k + 2$ to
  $3 \cdot 2^k + 1$ are related to irreducible choice sequences of lengths $2^{k+1} + 2$ to $3 \cdot 2^{k+1} + 1$.
  Indeed, these choice sequences look identical: the latter ones are just ``blown up'' by a factor of $2$. Since the
  substitution $\tau$ also ``blows up'' words by a factor of $2$, we proceed to look at $\tau$-images of the strategy
  trees of short choice sequences.

  Consider the strategy tree for the choice sequence $2212$ depicted in \autoref{fig:thue-morse_strategy}. Substitute
  all letters of this tree with $\tau$ while preserving the branch structure to obtain the right tree of
  \autoref{fig:thue-morse_strategy}. The obtained strategy tree gives a winning strategy for Alice in a word game with
  choice sequence $21211121$. Let us next give an intuitive explanation for the strategy from Alice's point of view.
  Alice can beat Bob in the word game with choice sequence $21211121$ by imagining that she plays the word game with
  choice sequence $2212$, for which she has a winning strategy. On her first turn, Alice lets Bob choose between $0$
  and $1$. Since Alice wins this game of length $1$, Alice can also win the game of length $2$ with choice sequence
  $21$ played on the $\tau$-images $\tau(0)$ and $\tau(1)$ (choice sequence $11$ is also possible but less
  interesting). Continuing, Alice lets Bob again choose between $0$ and $1$. The win on this play of length $2$ ensures
  Alice winning the game of length $4$ with choice sequence $2121$ played on the $\tau$-images $\tau(00)$, $\tau(01)$,
  $\tau(10)$, $\tau(11)$. Next, Alice gives Bob only one choice to ensure a win, so Bob, having no options, loses in
  the game of length $6$ with choice sequence $212111$ played on the respective $\tau$-images. Overall, we see that the
  short winning strategy for the choice sequence $2212$ enables Alice to always win the game with choice sequence
  $21211121$. This longer choice sequence is constructed in such a way that all occasions of Bob having a real choice
  (branches of the strategy tree) correspond to Bob having a choice of two letters in the shorter game with choice
  sequence $2212$; Alice just imagines playing a short game with choice sequence $2212$ filling the suffixes of the
  $\tau$-images by not letting Bob choose. Alice's method can indeed be viewed as a branch-preserving substitution of
  the strategy tree.

  The method described above does not explain if it is possible for Alice to obtain a winning strategy for, e.g., the
  choice sequence $2211121$ from some shorter winning strategy. Let us see how she could do this. Alice again imagines
  playing the winning strategy of the word game with choice sequence $2212$ using her winning strategy of
  \autoref{fig:thue-morse_strategy}. Now, however, during the first turn Alice lets Bob pick a suffix of length $1$ of
  the $\tau$-images of the letters $0$ and $1$ (which Bob is allowed to play on the first turn of the shorter game).
  Continuing as above, the played word will be a suffix of a word played in the word game with choice sequence
  $21211121$ and a suffix of a $\tau$-image of a word played in the word game with choice sequence $2212$. Therefore
  also $2211121 \in W(\tau)$. Similarly the play on the $\tau$-images does not have to complete the final image, the
  play can be restricted to a proper prefix of the $\tau$-images. In this particular case of the Thue-Morse
  substitution, it is easy to be convinced that all long enough winning strategies are obtainable by substitution by
  working out some example desubstitutions on strategy trees.

  In the next section, we will prove that the above methods always produce longer winning strategies from short
  winning strategies, even in the case of a general uniform substitution. We will show that not all long enough
  winning strategies are necessarily obtainable from short ones by substitution, but we will show that this holds for
  marked uniform substitutions. In essence, Alice can derive winning strategies for all long enough choice sequences in
  $W(\tau)$ from a few core strategies. Moreover, we are able to deduce the first difference function of a marked
  uniform substitution, which makes it possible to derive a formula for the factor complexity function.
  
  Knowing that winning strategies are obtained by substitution is not enough to give a complete description of the
  winning shift $W(\tau)$. There is typically some ambiguity on short prefixes of words in $W(\tau)$ due to the fact
  that they are related to the winning sets of word games played on suffixes of $\tau$-images. The winning sets of
  proper suffixes of $\tau$-images of a marked substitution can be very complicated---nothing general can be stated
  about their form. Thus at the end of \autoref{sec:main_results}, we introduce additional assumptions that simplify
  these winning sets. We show that the winning sets of proper suffixes of the $\tau$-images of permutive uniform
  substitutions are trivial, so that $W(\tau)$ admits a complete description. In this case, it can be shown that also
  the winning shift $W(\tau)$, not only the winning strategies, has a substitutive structure.

  Let us conclude this section by describing the substitutive structure of $W(\tau)$ in our example case of the
  Thue-Morse substitution. Let $\sigma$ be a substitution defined by $\sigma(1) = 11$ and $\sigma(2) = 21$, and let
  $\diamond w 2$ be an irreducible choice sequence in $W(\tau)$ for a letter $\diamond$. The result is that the words
  $\diamond \sigma(w) 2$ and $\sigma(\diamond w) 2$ are in $W(\tau)$ and that all irreducible choice sequences of
  length at least $5$ are obtained in this manner. Thus in our particular example it is sufficient to know all
  irreducible choice sequences of $W(\tau)$ of length at most $4$ to completely describe $W(\tau)$.
  
  \section{Main Results}\label{sec:main_results}
  In general, for a uniform substitution $\tau\colon S^* \to S^*$, substituting short winning strategies yields longer
  winning strategies in a manner similar to what was outlined in the previous section. To figure out the longer
  choice sequence obtained from a substituted short winning strategy, we need to identify the positions of the
  $\tau$-images $\tau(A)$ of a subset $A$ of $S$ where Bob can make choices without compromising the chances of Alice
  winning; in other words, we need to identify the winning set of $\tau(A)$. Notice that in general we obtain many
  possible choice sequences as the winning set of $\tau(A)$ might contain several words. We also want to consider the
  winning sets of prefixes and suffixes of these $\tau$-images since we want to include plays where in the beginning
  Bob plays a proper suffix of a $\tau$-image of a letter and in the end he plays a proper prefix of a $\tau$-image of
  a letter, just like in the examples of the previous section. Throughout this section, we assume that $\tau\colon S^*
  \to S^*$ is a uniform and aperiodic substitution of length $M$ with synchronization delay $L$.
  
  Before formalizing the ideas in the following lemma, we introduce some notation. Let $s$ be Alice's strategy for a
  word game $G$. We define its \emph{language} $\Lang{s}$ to consist of all possible plays with this strategy, that is,
  it is the set containing all words $p(G, s, s_B)$ for Bob's strategies $s_B$. Here, we let $\Lang[s]{n}$ denote the
  set $\pref_n(\Lang{s})$, that is, $\Lang[s]{n}$ contains the words that are playable in $n$ rounds when Alice uses
  the strategy $s$.

  \begin{lemma}\label{lem:short_to_long}
    Let $s$ be Alice's winning strategy for a word game $(S, n, \Lang{\tau}, \alpha)$ with $n \geq 2$. Then
    \begin{equation*}
      W(\suff_i(\tau(s(\varepsilon))))
      \cdot \prod_{k = 1}^{n - 2} \bigcap_{a \in \Lang[s]{k}} W(\tau(s(a)))
      \cdot \bigcap_{a \in \Lang[s]{n-1}} W(\pref_j(\tau(s(a)))) \subseteq W(\tau)
    \end{equation*}
    for all integers $i$ and $j$ such that $1 \leq i,j \leq M$.
  \end{lemma}
  \begin{proof}
    Let $\beta$ be in the set on the left side of the inclusion in the statement of the lemma. Notice that this set is
    indeed nonempty as the intersected sets all contain the word $1^M$ or $1^j$. We can factorize $\beta$ as
    $\beta_0 \beta_1 \cdots \beta_{n-1}$, where $\abs{\beta_0} = i$, $\abs{\beta_{n-1}} = j$, and $\abs{\beta_k} = M$
    for $1 \leq k < n - 1$. We define a strategy $s'$ for Alice for the word game $(S, i+(n-2)M+j, \Lang{\tau}, \beta)$
    as follows:
    \begin{itemize}
      \item first Alice plays according to a winning strategy for the game
            $(S, i, \suff_i(\tau(s(\varepsilon))), \beta_0)$ (such a strategy exist as $\beta_0$ was chosen to be in
            the winning set of $\suff_i(\tau(s(\varepsilon)))$);
      \item after $i + rM$ rounds have been played, Alice plays according to a winning strategy for the game
            $(S, M, \tau(s(a)), \beta_{r+1})$, where $a$ is a word in $\Lang[s]{r+1}$ such that $\tau(a)$ has the word
            of length $i + rM$ played so far as a suffix (the winning strategy exists because $\beta_{r+1}$ is in
            $W(\tau(s(a)))$ for all $a \in \Lang[s]{r+1}$);
      \item finally, after $i + (n - 2)M$ rounds, Alice plays according to a winning strategy for the game
            $(S, j, \pref_j(\tau(s(a))), \beta_{n-1})$, where $a$ is a word in $\Lang[s]{n-1}$ such that $\tau(a)$ has
            the word of length $\smash[t]{i + (n-2)M}$ played so far as a suffix (again, the winning strategy exists
            because $\beta_{n-1}$ is in $W(\pref_j(\tau(s(a))))$ for all $a \in \Lang[s]{n-1}$).
    \end{itemize}
    The described procedure clearly defines a strategy for Alice. What is left is to prove that the strategy $s'$ is a
    winning strategy for Alice in order to conclude that $\beta \in W(\tau)$.

    We show that Bob cannot produce a forbidden word during any round. During the first $i$ rounds Alice plays
    according to a winning strategy for the word game $(S, i, \suff_i(\tau(s(\varepsilon))), \beta_0)$, so a forbidden
    word cannot be produced. Suppose then that $i + rM$ rounds have been played without producing a forbidden word. The
    word played so far is a suffix of length $i + rM$ of the word $\tau(a_0 a_1 \cdots a_r)$ where
    $a_0 a_1 \cdots a_r \in \Lang{\tau}$. Alice plays next according to a winning strategy for the word game
    $(S, M, \tau(s(a_0 a_1 \cdots a_r)), \beta_{r+1})$, so the word played during the first $i + (r+1)M$ rounds is a
    suffix of length $i + (r+1)M$ of the word $\tau(a_0 a_1 \cdots a_{r+1})$ for some
    $a_{r+1} \in s(a_0 a_1 \cdots a_r)$. Since $s$ is a winning strategy, we see that
    $a_0 a_1 \cdots a_{r+1} \in \Lang{\tau}$, so also $\tau(a_0 a_1 \cdots a_{r+1}) \in \Lang{\tau}$. This means that
    no forbidden words are produced during the first $i + (r+1)M$ rounds. Similarly we see that no forbidden word is
    produced during the final $j$ rounds. We conclude that $s'$ is a winning strategy for Alice.
  \end{proof}
  
  \begin{example}\label{ex:no_desubstitution}
    In general, not all choice sequences in $W(\tau)$ are obtainable from shorter ones as in
    \autoref{lem:short_to_long}. Consider for instance the left-marked substitution
    \begin{equation*}
      \tau\colon
      \begin{array}{l}
        0 \mapsto 001 \\
        1 \mapsto 120 \\
        2 \mapsto 201
      \end{array}
    \end{equation*}
    with synchronization delay $5$.\footnote{This is quite tedious to find by hand, we used a computer.} Its fixed
    point is
    \begin{equation*}
      001001120001001120120201001001001120001001120120201001120201001201001 \cdots.
    \end{equation*}
    The left strategy of \autoref{fig:not_long_to_short} is a winning strategy of Alice for the choice sequence
    $12111111111112$. Let us show that this strategy is not obtainable from a shorter strategy by substitution. If it
    would be the case then, by desubstituting the words on the four paths of the strategy tree, we would obtain a
    winning strategy for Alice. This desubstituted strategy is depicted on the right in
    \autoref{fig:not_long_to_short}. The letter $\diamond$ stands for one of the letters $0$, $1$, and $2$; as $\tau$
    is not right-marked, it is not immediately obvious what $\diamond$ should be. Consider the words $\diamond 01201$
    and $\diamond 01202$ corresponding to the two top paths of this desubstituted tree. It is straightforward to see
    that in $\Lang{\tau}$ the factor $01201$ is extended to the left only by the letter $0$, but $01202$ is not
    extended to the left by $0$. This means that there is no choice for $\diamond$, so no desubstituted strategy is
    winning for Alice. Observe that this happens essentially due to the fact that the $\tau$-images of $0$ and $2$ have
    a common suffix of length $2$. Notice also that the right tree of \autoref{fig:not_long_to_short} corresponds by
    its branch structure to the choice sequence $21112$, which can checked not to be in $W(\tau)$.
  \end{example}

  \begin{figure}
    \centering
    \begin{minipage}{.5\textwidth}
      \centering
      \begin{tikzpicture}
        \tikzset{level distance=40pt}
        \tikzset{level 2/.style={level distance=90pt}}
        \Tree [.$1$ [.$001120201001$ [.$1$ ] [.$2$ ] ] [.$120001001120$ [.$0$ ] [.$1$ ] ] ]
      \end{tikzpicture}
    \end{minipage}%
    \begin{minipage}{.5\textwidth}
      \centering
      \begin{tikzpicture}
        \tikzset{level distance=40pt}
        \tikzset{level 2/.style={level distance=50pt}}
        \Tree [.$\diamond$ [.$0120$ [.$1$ ] [.$2$ ] ] [.$1001$ [.$0$ ] [.$1$ ] ] ]
      \end{tikzpicture}
    \end{minipage}
    \caption{Example of a long strategy that cannot be desubstituted into a short strategy.}\label{fig:not_long_to_short}
  \end{figure}
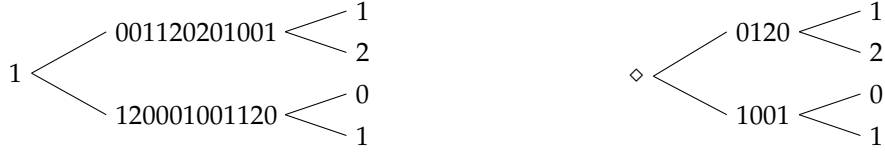

  Next we turn our attention to substitutions whose winning shifts consist essentially only of choice sequences as in
  \autoref{lem:short_to_long}. We begin with a definition.

  \begin{definition}
    Let $\alpha$ in $W(\tau)$ be a (finite) choice sequence such that $\abs{\alpha} > L$. If the winning strategies of
    $\alpha$ are obtainable from the winning strategies of shorter choice sequences in $W(\tau)$ by substitution as in
    \autoref{lem:short_to_long}, then we call $\alpha$ \emph{substitutive}.
  \end{definition}

  Our first step towards desubstituting long enough winning strategies is to consider left-marked substitutions for
  which we can prove the following lemma.
  
  \begin{lemma}\label{lem:left-marked_decomposition}
    Suppose that $\tau$ is left-marked. Let $\alpha$ in $W(\tau)$ be an irreducible choice sequence such that
    $\abs{\alpha} > L$. Then all winning plays of the game $(S, \abs{\alpha}, \Lang{\tau}, \alpha)$ have decomposition
    $\abs{\alpha} - 1 \bmod M$.
  \end{lemma}
  \begin{proof}
    Let $s$ be any winning strategy for Alice for the word game $(S, \abs{\alpha}, \Lang{\tau}, \alpha)$. We will prove
    that the last branching at the end of the strategy tree of $s$ marks a synchronization point of any winning play,
    that is, we claim that all winning plays by Alice with strategy $s$ have decomposition $\abs{\alpha} - 1 \bmod M$.
    Let $ua$ be a word in $\Lang{s}$ for some word $u$ and letter $a$, and suppose that $ua$ has decomposition
    $i \bmod M$ (the decomposition is well-defined as $\abs{ua} \geq L$). Let $r$ be the largest integer such that
    $rM < \abs{\alpha} - i$. Consider the suffix $v$ of $u$ of length $\abs{\alpha} - rM - i - 1$, so that $va$ is a
    prefix of $\tau(c)$ for some $c \in S$. Since $\alpha$ is irreducible, the word $ub$ is a winning play for some
    letter $b$ such that $a \neq b$. Now the word $ub$ must also have decomposition $i \bmod M$ as otherwise deleting
    the last letter from the words $ua$ and $ub$ would yield two different decompositions $\bmod M$ for the word $u$
    contradicting the assumption $\abs{u} \geq L$. Thus by repeating the preceding arguments, we see that $vb$ is a
    prefix of $\tau(d)$ for some $d \in S$. Since $\tau$ is left-marked, the only option is that $v$ is empty.
    Consequently, we have $i \equiv \abs{\alpha} - 1 \pmod{M}$. Since $s$ was an arbitrary winning strategy, the claim
    follows.
  \end{proof}

  \begin{example}
    Continuing \autoref{ex:no_desubstitution}, consider the winning plays of the word game with choice sequence
    $12111111111112$, depicted on the left in \autoref{fig:not_long_to_short}. All four possible plays
    $10011202010011$, $10011202010012$, $11200010011200$, and $11200010011201$ indeed have decomposition
    $14 - 1 \equiv 1 \pmod{3}$.
  \end{example}

  \autoref{lem:left-marked_decomposition} lets us define the notion of decomposition $\bmod M$ for long enough
  irreducible choice sequences.

  \begin{definition}
    Suppose that $\tau$ is left-marked, and let $\alpha$ be an irreducible choice sequence in $W(\tau)$ such that
    $\abs{\alpha} > L$. We say that $\alpha$ has \emph{decomposition $i \bmod M$} where $i$ is the unique number such
    that all winning plays of the game $(S, \abs{\alpha}, \Lang{\tau}, \alpha)$ have decomposition $i \bmod M$.
  \end{definition}

  \begin{example}
    Let us show that without assuming that $\tau$ is left-marked, the claim of \autoref{lem:left-marked_decomposition}
    is not always true. Consider the primitive substitution
    \begin{equation*}
      \tau\colon
      \begin{array}{l}
        0 \mapsto 021 \\
        1 \mapsto 010 \\
        2 \mapsto 210,
      \end{array}
    \end{equation*}
    which is not left-marked, nor are any of its conjugates since $\tau$ does not have any.\footnote{The conjugate of a substitution $\tau$ is the substitution obtained by cyclically shifting the common prefix of the $\tau$-images. A substitution and its conjugate have the same language.}
    The substitution $\tau$ has synchronization delay $6$, and its fixed point is
    \begin{equation*}
      021210010210010021021010021210010021021010021021210010021210010021010 \cdots.
    \end{equation*}
    The strategy tree of \autoref{fig:decomposition_fail} shows that $3111112 \in W(\tau)$. Now not all plays with this
    winning strategy have the same decomposition $\bmod{3}$ because in the $\tau$-images only two distinct letters may
    occur at a fixed position. In fact, we conjecture something stronger: $31^n 2 \in W(\tau)$ for infinitely many $n$.
  \end{example}

  \begin{figure}
    \centering
    \begin{tikzpicture}
      \tikzset{level 1/.style={level distance=60pt}}
      \tikzset{level 2/.style={level distance=80pt}}
      \Tree [.$\varepsilon$ [.$010021$ [.$0$ ] [.$2$ ] ] [.$100210$ [.$1$ ] [.$2$ ] ] [.$210010$ [.$0$ ] [.$2$ ] ] ]
    \end{tikzpicture}
    \caption{Example of a long winning strategy whose plays have different decompositions $\bmod{3}$.}\label{fig:decomposition_fail}
  \end{figure}
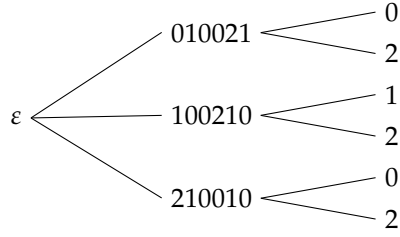

  Before we begin desubstituting long strategies, we prove the following lemma, which gives a description of
  the form of the choice sequences in $W(\tau)$. Let
  $\sigma_i\colon \{1,2,\ldots,\abs{S}\}^* \to \{1,2,\ldots,\abs{S}\}^*$ be the substitution defined by
  $\sigma_i(k) = k1^{i-1}$ for $k \in \{1,2,\ldots,\abs{S}\}$.

  \begin{lemma}\label{lem:left-marked_form}
    Suppose that $\tau$ is left-marked. If $\alpha$ is an irreducible choice sequence in $W(\tau)$ such that
    $\abs{\alpha} = rM + i + 1 > L$ with $0 \leq i < M$ ($\alpha$ has decomposition $i \bmod{M}$), then
    $\del_{i,1}(\alpha) \in \sigma_M(\{1,2,\ldots,\abs{S}\}^r)$.
  \end{lemma}
  \begin{proof}
    If $r = 0$, then there is nothing to prove, so we assume that $r > 0$. Consider the positions $\abs{\alpha}-M-1$,
    $\abs{\alpha}-(M-1)-1$, \ldots, $\abs{\alpha}-2$ of $\alpha$. Among these positions only the position
    $\abs{\alpha}-M-1$ may contain a letter that is greater than $1$. Otherwise in some play Bob could make a choice
    inside a $\tau$-image; recall that the decomposition $\bmod{M}$ of the plays is fixed before the game even starts,
    see \autoref{lem:left-marked_decomposition}. This is impossible as $\tau$ is left-marked. Thus the letters at
    positions $\abs{\alpha}-M-1$ to $\abs{\alpha}-2$ spell out a word of the form $k1^{M-1}$ with
    $k \in \{1,2,\ldots,\abs{S}\}$. Thus by repeating this argument $r - 1$ more times, the claim follows.
  \end{proof}

  Next we consider only marked substitutions and show that then desubstitution is possible.

  \begin{theorem}\label{thm:marked_substitutive}
    Suppose that $\tau$ is marked. Let $\alpha$ in $W(\tau)$ be an irreducible choice sequence such that
    $\abs{\alpha} > L$. Then $\alpha$ is substitutive and if $\alpha$ has decomposition $i \bmod{M}$, then there exists
    an irreducible choice sequence $a_0 a_1 \cdots a_{n-1}$ in $W(\tau)$ with winning strategy $s$ such that
    \begin{equation*}
      \alpha \in W(\suff_i(\tau(s(\varepsilon))))\sigma_M(a_1 \cdots a_{n-2})a_{n-1}
    \end{equation*}
    where $n$ is the largest integer such that $(n-2)M < \abs{\alpha} - i$.
  \end{theorem}
  \begin{proof}
    Let $\alpha$ in $W(\tau)$ be an irreducible choice sequence having decomposition $i \bmod{M}$ such that
    $\abs{\alpha} > L$. Let $s$ be Alice's winning strategy for the word game with choice sequence $\alpha$. By
    definition, the strategy tree of $s$ branches at positions where $\alpha$ contains a letter that is greater than
    $1$. Let us show how to perform a branch-preserving desubstitution on $s$ to obtain a shorter winning strategy
    $s'$.
    
    Consider first a leaf of the strategy tree of $s$. Since $i \equiv \abs{\alpha} - 1 \pmod{M}$ by
    \autoref{lem:left-marked_decomposition}, the last letter $a$ of the play corresponding to this leaf is the first
    letter of some $\tau$-image. Since $\tau$ is left-marked, there is a unique letter $b$ such that $\tau(b)$ begins
    with $a$. We replace the leaf corresponding to $a$ with a leaf corresponding to $b$.

    Next we show how to desubstitute the factors between two branchings in the middle of the strategy tree $s$. Say $j$
    and $k$ are consecutive positions of $\alpha$ containing letters that are greater than $1$ such that
    $k > j \geq i$. By \autoref{lem:left-marked_form}, the factor of $\alpha$ starting at position $j$ and ending at
    position $k - 1$ is of the form $\ell 1^{tM - 1}$ for some $\ell \in \{2,\ldots,\abs{S}\}$. Let $w$ be any winning
    play with the strategy $s$. Since $w$ has decomposition $i \bmod M$, it follows that the factor of $w$ starting at
    position $j$ and ending at position $k - 1$ is a $\tau$-image of some shorter word in $\Lang{\tau}$. This means
    that after $i$ rounds have been played, any time Alice's strategy branches, Bob has just completed a $\tau$-image
    on his previous turn. This means that it is possible to do a branch-preserving desubstitution on the subtrees of
    length $\abs{\alpha} - i$ of the strategy tree of $s$: the factor played between two branchings is a $\tau$-image
    of a shorter word in $\Lang{\tau}$ and can be directly desubstituted (since $\tau$ is injective). If there are no
    branchings before the final branching, then we can directly desubstitute the factor of any play starting at
    position $i$ and ending at position $\abs{\alpha} - 2$ (which could be empty).
    
    Now if $i = 0$, then we have desubstituted the whole strategy tree of $s$, and we are done. Suppose that $i > 0$.
    As $\tau$ is right-marked, the letter at position $i - 1$ of $w$ uniquely identifies a letter $a$ in $S$ such that
    the prefix of $w$ of length $i$ is a suffix of $\tau(a)$. We modify $s$ by replacing the first $i$ choices by a
    single choice of $a$ on a path corresponding to the play $w$. In other words, we let $a \in s'(\varepsilon)$ and
    set $s'(a)$ to contain the desubstituted subtree obtained above for the suffix of $w$ of length $\abs{\alpha} - i$.
    Now $s'$ is a strategy and it has the same branch structure as $s$ save for the initial part of $i$ rounds. By
    construction, all plays by $s'$ are ancestors of the plays with the winning strategy $s$, so $s'$ must also be a
    winning strategy. The strategy $s$ is clearly obtained from the strategy $s'$ by substitution as in
    \autoref{lem:short_to_long}. Therefore $\alpha$ is substitutive. The desubstitution process described clearly
    indicates that $\alpha$ has the claimed form.
  \end{proof}

  The essential message of \autoref{thm:marked_substitutive} is that knowing all winning strategies for irreducible
  choice sequences in $W(\tau)$ up to length $L$ is enough to derive winning strategies for all irreducible choice
  sequences---Alice does not need to learn much to beat Bob. Notice also that we can effectively enumerate $W(\tau)$
  when $\tau$ is marked, the sets $W(\pref_i(\tau(s(\varepsilon))))$ in the statement of
  \autoref{thm:marked_substitutive} are easily found by exhaustive search.

  Notice that substituting a strategy tree by $\sigma_M$ preserves its branch structure. Conversely, desubstituting, as
  in \autoref{thm:marked_substitutive}, preserves most of the branch structure. Indeed, supposing that $\tau$ is
  marked, then the subtree of the winning strategy of a word in $W(\tau)$, as in the third paragraph of the proof of
  \autoref{thm:marked_substitutive}, has the same branch structure as the desubstituted subtree. The initial part of
  the tree comes from a winning set played on suffixes of $\tau$-images. As there are finitely many of these, we
  conclude that there can be only finitely many different branch structures in the winning trees associated to the
  winning shift $W(\tau)$.  This means that in any choice sequence the number of letters greater than $1$ is bounded.
  In essence, Bob can almost never make a difference: on most turns, he has no options but to play what Alice wants.
  Compared to real life games, this makes our game somewhat degenerate. We emphasize that a priori it is not clear if
  Bob gets to play often or not.

  Observe that substituting two short winning strategies for two distinct choice sequences of the same length could
  yield the same longer choice sequence. For instance, if $2u$ and $3u$ are in $W(\tau)$, then cutting a branch of
  length $\abs{u}+1$ from the winning strategy $s$ for the choice sequence $3u$ yields a winning strategy $s'$ for the
  choice sequence $2u$. It follows that $W(\suff_i(\tau(s'(\varepsilon)))) \subseteq W(\suff_i(\tau(s(\varepsilon))))$,
  so all choice sequences obtained by substituting the winning strategy $s'$ are already obtained by substituting the
  winning strategy $s$. This is further elaborated in the proof of \autoref{thm:first_difference}. Moreover, it is
  possible that by substituting two distinct winning strategies for a fixed choice sequence produces distinct long
  choice sequences.

  Notice that the prefix of $\alpha$ of length $i$, as in the statement of \autoref{thm:marked_substitutive}, can be
  very complicated: we only assume that $\tau$ is aperiodic and marked and that it has synchronization delay, so the
  interior parts of the $\tau$-images can be chosen almost arbitrarily. To simplify the situation, assume that $\tau$
  is permutive. It is now clear that the suffix games related to the $\tau$-images are trivial:
  $W(\suff_i(\tau(A))) = \{k1^{i-1}, (k-1)1^{i-1}, \ldots, 1^i\}$, where $A$ is a subset of $S$ of $k$ elements. To put
  it in other words: $W(\suff_i(\tau(A))) = \sigma_i(\{k,k-1,\ldots,1\})$. Thus by \autoref{thm:marked_substitutive},
  we see that the winning shift $W(\tau)$ has the following substitutive structure.

  \begin{proposition}\label{prp:permutive}
    Suppose that $\tau$ is permutive. If $\alpha$ in $W(\tau)$ is an irreducible choice sequence such that
    $\alpha = \diamond w a$ with letters $\diamond$ and $a$, then $\sigma_i(\diamond)\sigma_M(w)a$ is in $W(\tau)$ for
    $1 \leq i \leq M$, and all choice sequences $\alpha$ of length at least $L+1$ are obtained in this way.
  \end{proposition}
  
  Since $\sigma_i$ is injective, the relation of the preceding proposition is a bijection from irreducible choice
  sequences of length $\abs{\alpha}$ to irreducible choice sequences of length $i+(\abs{\alpha}-2)M+1$. Such a
  bijection exists also in the case where $\tau$ is only marked as we shall see next in \autoref{thm:first_difference}.
  For its proof, we need the following lemma.

  \begin{lemma}\label{lem:first_choices_unique}
    Let $ku \in W(X)$ for a set $X$, a letter $k$, and a word $u$, and suppose that $k$ is maximal (for $u$). Then
    there exists a unique subset $A$ of $S$ of size $k$ such that $s(\varepsilon) \subseteq A$ for all Alice's winning
    strategies $s$ for a choice sequence $tu$ with $0 \leq t \leq k$.
  \end{lemma}
  \begin{proof}
    Let $s$ and $s'$ be two different winning strategies for the choice sequence $ku$. If
    $s(\varepsilon) \neq s'(\varepsilon)$, then there would be a letter in, say,
    $s(\varepsilon) \setminus s'(\varepsilon)$. By removing the subtree of length $n$ associated to this letter from
    the strategy tree of $s$ and attaching it to the strategy tree of $s'$, we obtain a new strategy. This new strategy
    clearly is a winning strategy for Alice for the choice sequence $(k + 1)u$ contradicting the maximality of the
    letter $k$. Thus the set $s(\varepsilon)$ is the same for all Alice's winning strategies $s$ for the choice
    sequence $ku$, and we may denote it by $A$.

    Consider then a choice sequence $tu$ with $t < k$, and let $e$ be Alice's arbitrary winning strategy for it. It
    must be that $e(\varepsilon) \subseteq A$ as otherwise there would be a letter in $e(\varepsilon) \setminus A$, and
    we could attach the subtree associated to it to the strategy tree of Alice's winning strategy for the choice
    sequence $ku$, like above, to obtain a contradiction with the maximality of the letter $k$.
  \end{proof}

  The next theorem states the same result as \cite[Corollary~3]{1998:on_uniform_d0l_words}. For the statement, we
  define $K$ to be the least integer such that $MK + 1 \geq L$.

  \begin{theorem}\label{thm:first_difference}
    Assume that $\tau$ is marked. Suppose that $n \geq K + 2$, and write $n = M^k r + \ell + 1$ with $k \geq 0$,
    $r \in \{K, K + 1, \ldots, KM - 1\}$, and $\ell \in \{1, \ldots, M^k\}$. Then $\Delta(n) = \Delta(r + 2)$
  \end{theorem}
  \begin{proof}
    Consider irreducible choice sequences in $W(\tau)$ of length $n$ ending with a word $u$ of length $n - 1$. Let $k$
    be the largest letter such that $ku \in W(\tau)$. When a winning strategy for the choice sequence $ku$ is
    substituted, as in \autoref{lem:short_to_long}, we obtain a winning strategy for an irreducible choice sequence of
    length $n(i)$, where $n(i) = i + (n-2)M + 1$ with $1 \leq i \leq M$. Moreover, the final $(n-2)M + 1$ letters of
    such a choice sequence are independent of the prefix $k$ by \autoref{thm:marked_substitutive}. Further, as
    $n(i) > L$, \autoref{thm:marked_substitutive} implies that all irreducible choice sequences of length $n(i)$ are
    obtained by substitution. Now there are a total of $k$ irreducible choice sequences of length $n$ with suffix $u$,
    so if we show that a total of $k$ distinct irreducible choice sequences of length $n(i)$ are obtainable from them
    by substitution, then we have shown that there are equally many irreducible choice sequences of length $n$ and
    $n(i)$.

    Let $A$ be as in \autoref{lem:first_choices_unique}. Consider a choice sequence $tu$, $0 \leq t \leq k$, with
    winning strategy $s$. The choice sequences of length $n(i)$ obtained from $tu$ by substitution are determined by
    the words in $W(\suff_i(\tau(s(\varepsilon))))$. Lemmas \ref{lem:first_choices_unique} and \ref{lem:inclusion}
    imply that $W(\suff_i(\tau(s(\varepsilon)))) \subseteq W(\suff_i(\tau(A)))$, so what is relevant is the size of
    $W(\suff_i(\tau(A)))$. \autoref{lem:first_choices_unique} and \autoref{prp:cardinality} show that the size of
    $W(\suff_i(\tau(A)))$ is $k$. Therefore a total of $k$ irreducible choice sequences of length $n(i)$ are obtainable
    from choice sequences with suffix $u$. As mentioned in the previous paragraph, we have proved that
    $\Delta(n) = \Delta(n(i))$. The claim follows by a straightforward computation.
  \end{proof}

  \begin{example}
    \autoref{thm:first_difference} is not true if $\tau$ is only left-marked. Consider for instance the substitution
    $\tau$ of \autoref{ex:no_desubstitution} Now $14 = 3 \cdot 4 + 1 + 1$, so \autoref{thm:first_difference} would
    predict that $\Delta(14) = \Delta(6)$. However, by a direct computation, it can be seen that in this case
    $\Delta(14) = 5$ but $\Delta(6) = 4$.
  \end{example}

  \autoref{thm:first_difference} can be used to derive the factor complexity function $f$ of a marked uniform
  substitution $\tau$ because $f(n) = 1 + \sum_{i = 1}^n \Delta(i)$. As the precise details in finding the exact
  formula do not involve word games, we omit the details and refer the reader to
  \cite[Theorem~2]{1998:on_uniform_d0l_words}.

  Notice also that \autoref{thm:first_difference} proves that the first difference function is a $M$-automatic
  sequence, so the factor complexity function is a $M$-regular sequence; see \cite{2003:automatic_sequences}. This
  holds for arbitrary uniform substitution.

  \section{Winning Shifts of Generalized Thue-Morse Words}\label{sec:gtm}
  In this section, we describe the winning shifts of generalized Thue-Morse words and, using our results, derive the
  known formulas for their factor complexity functions. For more on generalized Thue-Morse words, see e.g.
  \cite{2000:sums_of_digits_overlaps_and_palindromes}. Our notation largely follows
  \cite{2012:generalized_thue-morse_words_and_palindromic_richness}.

  Let $s_b(n)$ denote the sum of digits in the base-$b$ representation of the integer $n$. For $b \geq 2$ and
  $m \geq 1$, the generalized Thue-Morse word $\infw{t}_{b,m}$ is defined as the infinite word whose letter at position
  $n$ equals $s_b(n) \bmod{m}$. It is straightforward to prove that $\infw{t}_{b,m}$ is the fixed point, beginning with
  the letter $0$, of the primitive substitution $\varphi_{b,m}$ defined by 
  \begin{equation*}
    \varphi_{b,m}(k) = k(k+1)(k+2) \cdots (k + (b-1)),
  \end{equation*}
  for $k \in \{0, 1, \ldots, m - 1\}$, where the letters are interpreted modulo $m$. The word $\infw{t}_{b,m}$ is
  ultimately periodic if and only if $b \equiv 1 \pmod{m}$ \cite{2000:sums_of_digits_overlaps_and_palindromes}. We make
  the assumption that $\infw{t}_{b,m}$ is aperiodic.
  
  To clarify the notation, from now on we assume that letters are elements of the group $\Z_m$, so that we can
  naturally add letters. Moreover, we keep $b$ and $m$ fixed and simply write $\varphi$ for $\varphi_{b,m}$.

  Let $\pi\colon \Z_m \to \Z_m$ denote the permutation defined by setting $\pi(k) = k + b - 1$. In other words, the
  permutation $\pi$ maps $k$ to the final letter of the word $\varphi(k)$. We set $q$ to be the order of $\pi$, that
  is, the least positive integer such that $q(b-1) \equiv 0 \pmod{m}$.

  To describe the winning shift $W(\varphi)$ of $\varphi$, it is crucial to know words of $\Lang{\varphi}$ of length
  $2$ and $3$. Our proof is almost verbatim from \cite{2012:generalized_thue-morse_words_and_palindromic_richness}.

  \begin{lemma}\label{lem:factors_2_3}
    We have
    \begin{itemize}
      \item $\Lang[\varphi]{2} = \{\pi^i(k-1)k \colon k \in \Z_m, 0 \leq i < q\}$ and
      \item $\Lang[\varphi]{3} = \{\pi^i(k-1)k(k+1) \colon k \in \Z_m, 0 \leq i < q\} \cup \{(k-1)k\pi^{-i}(k+1) \colon k \in \Z_m, 0 \leq i < q\}$.
    \end{itemize}
  \end{lemma}
  \begin{proof}
    Set $L_0 = \{(k-1)k \colon k \in \Z_m\}$. Clearly $L_0 \subseteq \Lang[\varphi]{2}$. Let $L_{j+1}$ to be the set of
    factors of length $2$ of the words in $\varphi(L_j)$. By the definition of $\pi$, we have
    $L_j = \{\pi^i(k-1)k \colon k \in \Z_m, 0 \leq i < j + 1\}$. Since $L_q = L_{q-1}$, we have
    $\Lang[\varphi]{2} = L_{q-1}$.

    By the form of $\varphi$, either the first two letters of a factor of length $3$ are equal or its last two letters
    are. The claim thus follows from the form of the factors of length $2$.
  \end{proof}

  The following lemma concerning the synchronization delay of $\varphi$ is proven in
  \cite{2012:factor_frequencies_in_generalized_thue-morse_words}; we repeat the proof here.

  \begin{lemma}\label{lem:synch_delay}
    The substitution $\varphi$ has synchronization delay $2b$.
  \end{lemma}
  \begin{proof}
    Consider a word $w$ of $\Lang{\varphi}$ of length $2b$. If $w$ contains a factor $k\ell$ with
    $\ell \neq k + 1$, then the factor $k\ell$ cannot occur inside a $\varphi$-image, so the position where $\ell$
    occurs marks a synchronization point. If such a factor does not occur in $w$, then the word $w$ is of the form
    $k(k+1) \cdots (k + 2b - 1)$, that is, $w = \varphi(k(k+b))$. Suppose for a contradiction that $w$ has ancestor
    $x_1 x_2 x_3$. Due to the form of $w$, we have $x_2 = x_1 + b$ and $x_3 = x_1 + 2b$, that is,
    $x_1 (x_1 + b) (x_1 + 2b) \in \Lang{\varphi}$. This is impossible as $x_1 + b \neq x_1 + 1$ and
    $x_1 + 2b \neq x_1 + b + 1$ due to our assumption that $b \neq 1$. Thus the only ancestor of $w$ is $k(k+b)$. We
    have thus shown that $L \leq 2b$.

    Fix $k \in \Z_m$. Since $k - b = k - 1 + (q-1)(b-1)$, we see that $(k-b)k \in \Lang{\varphi}$ by
    \autoref{lem:factors_2_3}. Consider the prefix $u$ of $\varphi((k-b)k)$ of length $2b - 1$. This
    prefix has $\varphi(k-1)$ as a suffix, and its prefix of length $b-1$ is a suffix of $\varphi(k-b-1)$.
    Because $(k-1-b)(k-1) \in \Lang{\varphi}$, the word $u$ has two ancestors proving that $L \geq 2b$.
  \end{proof}

  Since $\varphi$ is permutive, it now follows that every choice sequence in $W(\varphi)$ having length at least
  $2b + 1$ is obtainable by substitution from a shorter choice sequence. Next we describe the choice sequences of
  length at most $2b$.
  
  \begin{proposition}\label{prp:gtm_short}
    Let $\alpha$ in $W(\varphi)$ be an irreducible choice sequence of length $n$.
    \begin{enumerate}[(i)]
      \item If $2 \leq n \leq b + 1$, then $\alpha = \diamond 1^{n-2} a$ with $\diamond \in \{1,\ldots,m\}$ and
            $a \in \{2,\ldots,q\}$.
      \item If $b + 2 \leq n \leq 2b$, then $\alpha = \diamond 1^{n-2} a$ or $\alpha = \diamond 1^\ell 2 1^{b-1} 2$ with
            $\diamond \in \{1,\ldots,m\}$ and $a \in \{2,\ldots,q\}$.
    \end{enumerate}
    Moreover, each word of such form is in $W(\varphi)$.
  \end{proposition}
  \begin{proof}
    Consider first the case $2 \leq n \leq b + 1$. Write $\alpha = \diamond u r$ with letters $\diamond$ and $r$, and
    let $w$ be a winning play in the game with choice sequence $\alpha$. First we argue that the prefix of $w$ of
    length $n - 1$ is of the form $k (k + 1) \cdots (k+n-2)$ for some $k \in \Z_m$, that is, it equals
    $\varphi_{n-1,m}(k)$. If this were not the case, then this prefix equals $xijy$ for some words $x$ and $y$ and
    letters $i$ and $j$ such that $j \neq i + 1$. Thus $w$ has decomposition $\abs{xi} \bmod{b}$. Since $w$ is a
    winning play, Bob cannot choose inside a $\varphi$-image, and it must thus be that $\abs{jy}$ is a positive
    multiple of $b$. This is impossible as now $\abs{\alpha} > \abs{xijy} \geq b + 1$. Due to the restricted form of
    the prefix of $w$ of length $n - 1$, we see that Bob cannot make any choices between his first and last turns, so
    $\alpha = \diamond 1^{n-2} r$. Suppose for a contradiction that $r > q$. Now Bob can pick a letter $c$ such that
    $c \notin \{\pi^i(k+n-1) \colon 0 \leq i < q\}$. It follows that $\pi^{-1}(k+n-2)c$ is an ancestor of the played
    word $\varphi_{n-1,m}(k)c$. This is however a contradiction with \autoref{lem:factors_2_3}. Therefore $r \leq q$.
    It is now clear that any word of the form $\diamond 1^{n-2} r$ with $\diamond \in \{1,\ldots,m\}$ and
    $r \in \{2,\ldots,q\}$ is in $W(\varphi)$: after Bob has chosen $k$, Alice forces him to play $\varphi_{n-1,m}(k)$
    after which she lets him choose among the $q$ letters $c$ such that $\pi^{-1}(k+n-2)c$ is in $\Lang[\varphi]{2}$.

    Suppose then that $b + 2 \leq n \leq 2b$. If $\alpha$ contains exactly two letters that are greater than $1$, one
    at the beginning and one at the end, then $\alpha$ must again be of the form $\diamond 1^{n-2} a$ with
    $\diamond \in \{1,\ldots,m\}$ and $a \in \{2,\ldots,q\}$ (after Bob has chosen $k$, Alice forces him to play
    $\varphi_{n-b-1}(k)\varphi(\pi^{-1}(k+n-b-2)+1)$ after which she lets him choose among the $q$ letters $c$ such
    that $\pi^{-1}(k+n-b-2)(\pi^{-1}(k+n-b-2)+1)c \in \Lang{\varphi}$; see \autoref{lem:factors_2_3}). Otherwise write
    $\alpha = \diamond urvs$ with letters $\diamond$, $r$, and $s$ such that $r, s > 1$, and let $w$ again be a winning
    play in the game with choice sequence $\alpha$. Analogous to the arguments of the preceding paragraph, we see that
    $\abs{\alpha} > \abs{\diamond urv} \geq 2b+1$ unless the prefix of $w$ of length $\abs{u} + 1$ is of the form
    $\varphi_{\abs{u}+1,m}(k)$ for some $k \in \Z_m$. Again, we have $u = 1^{\abs{u}}$ and, further, $v = 1^{b-1}$.
    Assume for a contradiction that $r \geq 3$. After $\abs{u} + 1$ rounds Bob can choose a letter $c$ such that
    $c \notin \{k+\abs{u}+1, \pi^{-1}(k+\abs{u})+1\}$. Clearly the word played so far has decomposition
    $\abs{u}+1 \bmod{b}$, so during her next $b-1$ turns Alice must let Bob complete the $\varphi$-image beginning with
    $c$. During his final turn Bob can pick a letter $d$ such that $d \neq c+1$. It follows that the played word has
    the word $\pi^{-1}(k+\abs{u})cd$ as an ancestor. By \autoref{lem:factors_2_3}, this ancestor is not in
    $\Lang[\varphi]{3}$, so Bob wins. This is a contradiction, so $r = 2$. The preceding arguments also show that
    $w$ must have $\varphi_{\abs{u}+1,m}(k)(k+\abs{u}+1)$ or $\varphi_{\abs{u}+1,m}(k)(\pi^{-1}(k+\abs{u})+1)$ as a
    prefix. Let us consider the former case. Since Bob wins if he can choose inside a $\varphi$-image, Alice must now
    force Bob to play $\varphi_{n-1,m}(k)$ to ensure that the word played so far has multiple ancestors. If $s \geq 3$,
    then as his ultimate move Bob can pick a letter $c$ such that $c \notin \{k+n-1,\pi^{-1}(k+n-2)+1\}$. Then $w$ has
    unique ancestor $\pi^{-1}(k+n-2-b)\pi^{-1}(k+n-2)c$. Our assumption that $b \neq 1$ implies by
    \autoref{lem:factors_2_3} that $\pi^{-1}(k+n-2)+1 = c$, which is impossible by the choice of $c$. Thus $s = 2$,
    that is, $\alpha = \diamond 1^{\abs{u}} 21^{b-1}2$. It is now straightforward to derive a winning strategy for
    Alice for any $\diamond \in \{1,\ldots,m\}$. The subtree of length $n-1$ of such a strategy is depicted in
    \autoref{fig:gtm_subtree}; it is readily verified that the corresponding strategy is winning for Alice using
    \autoref{lem:factors_2_3}. The claim follows.
  \end{proof}

  \begin{figure}
    \centering
    \begin{tikzpicture}
      \tikzset{level 1/.style={level distance=80pt}}
      \tikzset{level 2/.style={level distance=190pt}}
      \Tree [.$k\cdots(k+\ell)$ [.$(k+\ell+1)\cdots(k+\ell+b)$ [.$k+\ell+b+1$ ] [.$k+\ell+2$ ] ] [.$(\pi^{-1}(k+\ell)+1)\cdots(\pi^{-1}(k+\ell)+b)$ [.$\pi^{-1}(k+\ell)+2$ ] [.$k+\ell+2$ ] ] ]
    \end{tikzpicture}
    \caption{The subtree of Alice's winning strategy after Bob has chosen $k$ in the game with choice sequence $\diamond 1^\ell 21^{b-1}2$.}\label{fig:gtm_subtree}
  \end{figure}
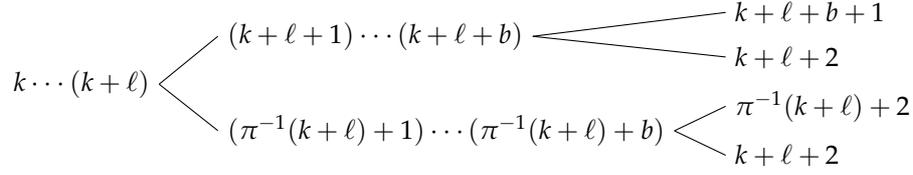

  Since $\varphi$ is permutive, all long enough choice sequences $\alpha$ in $W(\varphi)$ are of the form
  $\sigma_i(\diamond)\sigma_b(w)a$, where $\diamond wa \in W(\varphi)$ for letters $\diamond$ and $a$. Combining this
  with \autoref{prp:gtm_short}, we see that the winning shift $W(\varphi)$ indeed has the same form as described in
  \autoref{sec:tm_example}. Either $\alpha$ is of the form $\diamond 1^{\abs{\alpha}-2} a$ with
  $\diamond \in \{1,\ldots,m\}$ and $a \in \{2,\ldots,q\}$ or $\smash[t]{\alpha = \diamond 1^\ell 2 1^{b^k-1} 2}$,
  where $\diamond \in \{1,\ldots,m\}$, $k$ is the largest $k$ such that $b^k < \abs{\alpha}$ and
  $0 \leq \ell \leq b^k - b^{k-1} - 1$.

  \autoref{prp:gtm_short} together with \autoref{thm:first_difference} implies that for $n \geq 2$ the first difference
  function $\Delta(n)$ for $\infw{t}_{b,m}$ takes only two values: $(q-1)m$ and $qm$. Using induction, we can derive
  the values of $\Delta(n)$ and $C(n)$ (the factor complexity function of $\infw{t}_{b,m}$) for any $n \geq 1$; see
  \autoref{tbl:gtm_complexity}. These functions have been derived by \v{S}. Starosta with other methods
  \cite{2012:generalized_thue-morse_words_and_palindromic_richness}.

  {\renewcommand{\arraystretch}{1.3}
  \begin{table}
    \centering
    \begin{tabular}{|c|c|c|}
      \hline
      $n$                                                   & $\Delta(n)$               & $C(n)$ \\ \hline
      $1$                                                   & $m - 1$                   & $m$ \\ \hline
      $2 \leq n \leq b + 1$                                 & $(q-1)m$                  & $qm(n-1)-m(n-2)$ \\ \hline
      $b^{k+1} + \ell + 1$                                  & \multirow{2}{*}{$qm$}     & \multirow{2}{*}{$qm(n-1)-m(b^{k+1}-b^k)$} \\
      $k \geq 0, 1 \leq \ell \leq b^{k+1} - b^k$            & & \\ \hline
      $2b^{k+1} - b^k + \ell + 1$                           & \multirow{2}{*}{$(q-1)m$} & \multirow{2}{*}{$qm(n-1)-m(b^{k+1}-b^k+\ell)$} \\
      $k \geq 0, 1 \leq \ell \leq b^{k+2} - 2b^{k+1} + b^k$ & & \\ \hline
    \end{tabular}
    \caption{The values of the first difference function $\Delta(n)$ and the factor complexity function $C(n)$ of the generalized Thue-Morse word $\infw{t}_{b,m}$.}\label{tbl:gtm_complexity}
  \end{table}}

  \section*{Acknowledgments}
  The work of the first author was supported by the Finnish Cultural Foundation by a personal grant.

  \printbibliography
\end{document}
% vim: set textwidth=119: